\documentclass[runningheads]{llncs}
\usepackage{graphicx}
\usepackage{subcaption}
\usepackage{xcolor}
\usepackage{enumitem}
\usepackage[T1]{fontenc}
\usepackage{hyperref}

\usepackage{lipsum} 

\usepackage{tikz}
\usetikzlibrary{arrows,shapes}
\usetikzlibrary{calc}

\usepackage[linesnumbered,ruled,vlined]{algorithm2e}

\usepackage{amssymb}
\usepackage{stackrel}
\usepackage{mathtools}

\usepackage[dvips]{epsfig}
\usepackage[english,francais]{babel}
\usepackage[utf8]{inputenc}

\usepackage{bbm}
\usepackage{amsmath,amsfonts,amssymb}

\begin{document}
\title{Morse frames\thanks{This version of the paper corrects a minor error in the algorithm presented in Section~6 of the published previous version of this paper. The correction does not affect the main results.}}
\author{Gilles Bertrand\orcidID{0009-0004-7294-7081} \\\and Laurent Najman\orcidID{0000-0002-6190-0235}}%
\authorrunning{G. Bertrand and L. Najman}%
\institute{Univ Gustave Eiffel, CNRS, LIGM, F-77454 Marne-la-Vallée, France %
\email{\{gilles.bertrand,laurent.najman\}@esiee.fr}%
}
\maketitle              %

\newcommand{\bbbu}{\; \ddot{\cup} \;}
\newcommand{\axr}[1]{\ddot{\textsc{#1}}  \normalsize}

\newcommand{\axcup}{\textsc{C\tiny{UP}} }
\newcommand{\axcap}{\textsc{C\tiny{AP}} }
\newcommand{\axunion}{\textsc{U\tiny{NION}} }
\newcommand{\axinter}{\textsc{I\tiny{NTER}} }

\newcommand{\bb}[1]{\mathbb{#1}}
\newcommand{\ca}[1]{\mathcal{#1}}
\newcommand{\ax}[1]{\textsc{#1} \normalsize}

\newcommand{\axb}[2]{\ddot{\textsc{#1}} \textsc{\tiny{#2}}  \normalsize}
\newcommand{\bbb}[1]{\ddot{\mathbb{#1}}}
\newcommand{\cab}[1]{\ddot{\mathcal{#1}}}

\newcommand{\rel}[1]{\scriptstyle{\mathbf{#1}}}

\newcommand{\rela}[1]{\textsc{\scriptsize{\bf{{#1}}}} \normalsize}

\newcommand{\de}[2]{#1[#2]}
\newcommand{\di}[2]{#1\langle #2 \rangle}

\newcommand{\la}{\langle}
\newcommand{\ra}{\rangle}
\newcommand{\hs}{\hspace*{\fill}}

\newcommand{\cell}{\mathbb{C}}
\newcommand{\cellp}{\mathbb{C}^\times}
\newcommand{\simp}{\mathbb{S}}
\newcommand{\comp}{\mathbb{H}}
\newcommand{\simpp}{\mathbb{S}^\times}
\newcommand{\den}{\mathbb{D}\mathrm{en}}
\newcommand{\ram}{\mathbb{R}\mathrm{am}}
\newcommand{\tree}{\mathbb{T}\mathrm{ree}}
\newcommand{\graph}{\mathbb{G}\mathrm{raph}}
\newcommand{\vertex}{\mathbb{V}\mathrm{ert}}
\newcommand{\edge}{\mathbb{E}\mathrm{dge}}
\newcommand{\equ}{\mathbb{E}\mathrm{qu}}
\newcommand{\esub}{\mathbb{E}\mathrm{sub}}

\newcommand{\topp}{\langle \mathrm{K} \rangle}
\newcommand{\topq}{\langle \mathrm{Q} \rangle}
\newcommand{\topxp}{\langle \mathbb{X}, \mathrm{P} \rangle}
\newcommand{\topxq}{\langle \mathbb{X},\mathrm{Q} \rangle}

\newcommand{\vt}{\mathcal{K}}
\newcommand{\vtp}{\mathcal{T'}}
\newcommand{\vtpp}{\mathcal{T''}}
\newcommand{\vq}{\mathcal{Q}}
\newcommand{\vqp}{\mathcal{Q'}}
\newcommand{\vqpp}{\mathcal{Q''}}
\newcommand{\vk}{\mathcal{K}}
\newcommand{\vkp}{\mathcal{K'}}
\newcommand{\vkpp}{\mathcal{K''}}

\newcommand{\C}{\ensuremath{\searrow^{\!\!\!\!\!C}}}
\newcommand{\Detach}{\ensuremath{\;\oslash\;}}

\newcommand{\sig}{\sigma}
\newcommand{\del}{W}
\newcommand{\ms}{\overrightarrow{W}}
\newcommand{\msi}{\overrightarrow{W_{i}}}
\newcommand{\msim}{\overrightarrow{W_{i-1}}}
\newcommand{\mss}{\widehat{W}}
\newcommand{\msc}{\ddot{W}}

\newcommand{\torusms}{\overrightarrow{T}}
\newcommand{\torusmso}{\overleftarrow{T}}

\begin{abstract}
In the context of discrete Morse theory, we introduce Morse frames, which are maps that associate a set of critical simplexes to each simplex of a given complex. The main example of Morse frames are the Morse references. In particular,  Morse references allow computing Morse complexes, an important tool for homology. We highlight the link between Morse references and gradient flows. We also propose a novel presentation of the annotation algorithm for persistent cohomology, as a variant of a Morse frame. Finally, we propose another construction, that takes advantage of the Morse reference for computing the Betti numbers of a complex in mod 2 arithmetic.
\end{abstract}

\section{Introduction}

In this paper, we aim at developing new concepts and algorithm schemes for computing topological invariants for simplicial complexes, such as cycles, cocycles and Betti numbers (Sec.~\ref{sec:basic}).
In \cite{Bertrand2023MorseSequences}, one of the authors of the present paper, introduces a novel, sequential, presentation of discrete Morse theory \cite{For98}, termed {\em Morse sequences}. In section~\ref{sec:frame}, we introduce \emph{Morse frames}, that are maps that associate a set of critical simplexes to each simplex. These maps allow adding information to Morse sequences, so that we can compute cycles and cocycles that detect “holes”.
The main example of Morse frames is called the Morse reference (Sec.~\ref{sec:reference}), and is a by-product of Morse sequences. We discuss the link between reference maps and gradient flows. This leads us to the Morse complex.
We then see (Sec.~\ref{sec:annotations}) that Morse frames allows for a novel presentation of annotations~\cite{Dey14} for computing persistent cohomology. Then, inspired by the annotation technique, we propose (Sec.~\ref{sec:mergingReferenceAnnotations}) an efficient construction for computing Betti numbers in mod 2 arithmetic.
We then discuss (Sec.~\ref{sec:implementing}) how to implement the notions presented in the paper. Finally, we conclude the paper. 

\section{Simplicial complexes, homology, and cohomology}
\label{sec:basic}
\subsection{Simplicial complexes}
Let $K$ be a finite family composed
of non-empty finite sets.
The family $K$ is a {\it (simplicial) complex} if $\sigma \in K$ whenever $\sigma \not= \emptyset$ and $\sigma \subseteq \tau$
for some $\tau \in K$. 

An element of a simplicial complex $K$ is {\it a simplex of $K$}, or {\it a face of $K$}.
A {\em facet of $K$} is a simplex of $K$ that is maximal for inclusion.
The {\it dimension} of $\sigma \in K$, written $dim(\sigma)$,
is the number of its elements
minus one. If $dim(\sigma) =p$, we say that $\sigma$ is a \emph{$p$-simplex}.
We denote by $K^{(p)}$ the set %
of all $p$-simplexes of~$K$.

We recall the definitions of the collapses/expansions operators~\cite{Whi39}.

Let $K,L$ be simplicial complexes. Let $\sig \in K^{(p)}$, $\tau \in K^{(p+1)}$.
The couple $(\sig,\tau)$ is a {\em free pair for $K$}, or a {\em free $p$-pair for $K$},
if $\tau$ is the only face of $K$ that contains $\sig$.
Thus, $\tau$ is necessarily a facet of $K$.
If $(\sig,\tau)$ is a free ($p$-)pair for $K$, then $L = K \setminus (\sig,\tau)$
is {\em an elementary ($p$-)collapse of $K$}, and $K$ is {\em an elementary ($p$-)expansion of $L$}.
We say that
$K$ {\em collapses onto $L$},
or that $L$ {\em expands onto $K$},
if there exists a sequence
$\langle K=K_0,\ldots,K_k=L \rangle$, such that
$K_i$ is an elementary collapse of $K_{i-1}$, $i \in [1,k]$.

\subsection{Homology and cohomology}
\label{sec:hom}
Let $K$ be a simplicial complex. We write
$K[p]$ for the set composed of all subsets of $K^{(p)}$.
Also, we set $K^{(-1)} = \emptyset$ and $K[-1] = \{\emptyset\}$.
Each element of $K[p]$, $p \geq -1$, is a \emph{$p$-chain of $K$}.
The symmetric difference of two elements of $K[p]$ endows $K[p]$
with the structure of a vector space over the field $\mathbb{Z}_2 = \{ 0, 1 \}$. The set $K^{(p)}$
is a basis for this vector space.
Within this structure, a chain $c \in K[p]$ is written as a sum
$\sum_{\sigma \in c} \sigma$, the chain $c = \emptyset$
being written $0$. The sum of two chains is obtained using the modulo 2 arithmetic.

Let $K$ be a simplicial complex. As we are dealing with a finite simplicial complex, boundary and coboundary operators can be defined as operators on $K[p]$. 
 If $\sigma \in K^{(p)}$, with $p \geq 0$, we set: 

\noindent
\hspace*{\fill}
$\partial(\sigma) = \{\tau \in K^{(p-1)} \; | \; \tau \subset \sigma \}$ and 
$\delta(\sigma) = \{\tau \in  K^{(p+1)}\; | \; \sigma \subset \tau \}$.
\hspace*{\fill} 

\noindent
The \emph{boundary operator} $\partial_p : K[p] \rightarrow K[p-1]$, $p \geq 0$, is such that, for
each $c \in K[p]$, 
$\partial_p(c) = \sum_{\sigma \in c} \partial(\sigma)$, with $\partial_p(\emptyset) = 0$.

\noindent
The \emph{coboundary operator} $\delta^p : K[p] \rightarrow K[p+1]$, $p \geq -1$, is such that, for
each $c \in K[p]$, 
$\delta^p(c) = \sum_{\sigma \in c} \delta(\sigma)$, with $\delta^p(\emptyset) = 0$. 

\noindent
For each $p \geq 0$, we have
$ \partial_{p} \circ \partial_{p+1} = 0$ and $\delta^{p} \circ \delta^{p-1} = 0$. \\
We define four subsets of $K[p]$, $p \geq 0$, which are vector spaces over $\bb{Z}_2$: 
\begin{itemize}[noitemsep,topsep=0pt]
\item the set $Z_p(K)$ of \emph{$p$-cycles of $K$}, $Z_p(K)$ is the kernel of $\partial_p$; 
\item the set $B_p(K)$ of \emph{$p$-boundaries of $K$}, $B_p(K)$ is the image of $\partial_{p+1}$; 
\item the set $Z^p(K)$ of \emph{$p$-cocycles of $K$}, $Z^p(K)$ is the kernel of $\delta^p$; 
\item the set $B^p(K)$ of \emph{$p$-coboundaries of $K$}, $B^p(K)$ is the image of $\delta^{p-1}$. 
\end{itemize}

\def\figsize1{0.22}
\begin{figure*}[tb]
    \centering
        \begin{subfigure}[t]{\figsize1\textwidth}
        \centering
        \includegraphics[width=.99\textwidth]{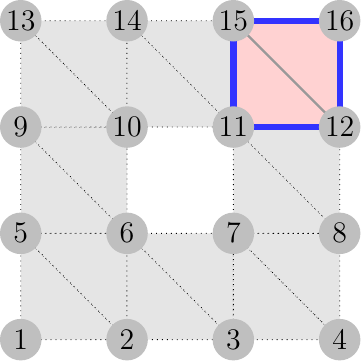}
        \caption{%
        }
    \end{subfigure}%
    ~
    \begin{subfigure}[t]{\figsize1\textwidth}
        \centering
        \includegraphics[width=.99\textwidth]{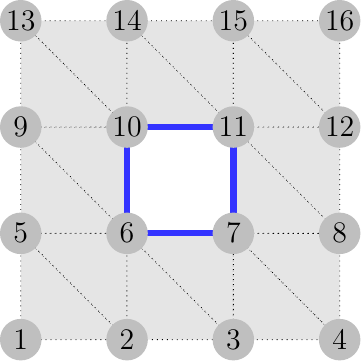}
        \caption{%
        }
    \end{subfigure}%
    ~
    \begin{subfigure}[t]{\figsize1\textwidth}
        \centering
        \includegraphics[width=.99\textwidth]{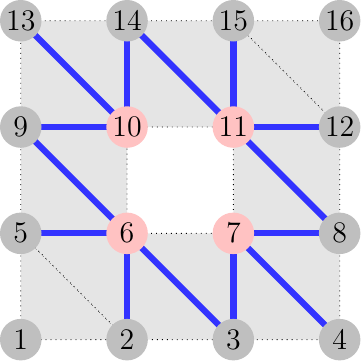}
        \caption{%
        }
    \end{subfigure}%
    ~ 
    \begin{subfigure}[t]{\figsize1\textwidth}
        \centering
        \includegraphics[width=.99\textwidth]{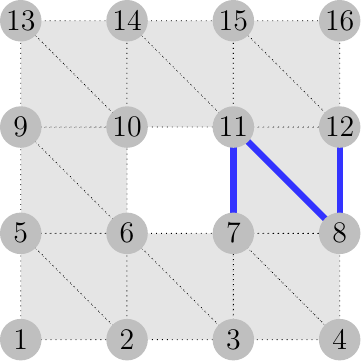}
        \caption{%
        }
    \end{subfigure}
    \caption{\label{fig:annulus} An annulus, with various cycles and cocyles. See text for details. %
}
\end{figure*}

Fig.~\ref{fig:annulus} depicts an annulus, with various cycles and cocyles, coloured in blue. In Fig.~\ref{fig:annulus}.a, we see a 1-cycle that is the 1-boundary of the two pink triangles.  In Fig.~\ref{fig:annulus}.b, we have a 1-cycle that is not a 1-boundary. Such a cycle detects a “hole” by “contouring” it.
In Fig.~\ref{fig:annulus}.c, we see
 a 1-cocycle which is the 1-coboundary of the four pink points. In Fig.~\ref{fig:annulus}.d, we have a 1-cocycle that is not a 1-coboundary. Such a cocycle detects a “hole” by “cutting” the annulus.

\noindent
We also define the following quotient vector spaces: 
\begin{itemize}[noitemsep,topsep=0pt]
\item $H_p(K) = Z_p(K) \setminus B_p(K)$, which is the \emph{$p^\textnormal{th}$ homology vector space of $K$}; 
\item $H^p(K) = Z^p(K) \setminus B^p(K)$, which is the \emph{$p^\textnormal{th}$ cohomology vector space of $K$}. 
\end{itemize}

\noindent
An element $h$ in $H_p(K)$  is such that $h = z + B_p(K)$ for some $z \in Z_p(K)$. We
write $h = [ z ]_p$, which  is the \emph{homology class of the cycle z}. 

\noindent
Similarly, an element $h$ in $H^p(K)$  is such that $h = z + B^p(K)$ for some $z \in Z^p(K)$. We
write $h = [ z ]^p$, which is the \emph{cohomology class of the cocycle z}.

Let $\beta_p(K) = dim(H_p(K))$ and $\beta^p(K) = dim(H^p(K))$.
We have $\beta_p(K) = \beta^p(K)$
(See \cite[Sec. V.1]{Edel01}).
The number $\beta_p(K) = \beta^p(K)$ is the \emph{$p^\textnormal{th}$ Betti number (mod 2) of $K$}.

\section{Morse sequences and Morse frames}
\label{sec:frame}
Let us first introduce the two following basic operators \cite{Whi39}.

Let $K,L$ be simplicial complexes.
If $\sigma \in K$ is a facet of $K$, and if $L = K \setminus \{\sigma \}$, we say that
$L$ is {\em an elementary perforation of $K$}, and that
$K$ is {\em an elementary filling of $L$}. 

The notion of a ``Morse sequence'' \cite{Bertrand2023MorseSequences} is defined by simply considering expansions and fillings of a simplicial complex.

\begin{definition} \label{def:seq1}
Let $K$ be a simplicial complex. A \emph{Morse sequence (on $K$)} is a sequence
$\ms = \langle \emptyset = K_0,\ldots,K_k =K \rangle$ of simplicial complexes such that,
for each $i \in [1,k]$, $K_i$ is either an elementary expansion or an elementary filling of $K_{i-1}$.
\end{definition}

Let $\ms = \langle K_0,\ldots,K_k \rangle$ be a Morse sequence.
For each $i \in [1,k]$: 
\begin{itemize}[noitemsep,topsep=0pt]
\item If $K_i$ is an elementary filling of $K_{i-1}$, we write $\hat{\sigma}_i$ for the simplex $\sig$ such that $K_i = K_{i-1} \cup \{\sig\}$. We say that
the face $\sig$ is \emph{critical for $\ms$}. 
\item If $K_i$ is an elementary expansion of $K_{i-1}$, we write
$\hat{\sigma}_i$  for the free pair $(\sigma,\tau)$ such that $K_i = K_{i-1} \cup \{\sigma,\tau \}$. We say that
$\hat{\sigma}_i$, $\sig$, $\tau$, are \emph{regular for $\ms$}. 
\end{itemize}

\noindent
We write $\mss = \langle \hat{\sigma}_1,\ldots, \hat{\sigma}_k\rangle$, and we say that
$\mss$ is a \emph{(simplex-wise) Morse sequence}.
Clearly, $\ms$ and $\mss$ are two equivalent forms.
We shall pass from one of these forms to the other without
notice.

There are several ways to obtain a Morse sequence $\ms$ from a given complex~$K$.
The two following schemes are basic ones to achieve this goal: 
\begin{enumerate}[noitemsep,topsep=0pt]
\item \emph{The increasing scheme}. We build $\ms$ from the left to the right. Starting from~$\emptyset$, we obtain $K$ by iterative expansions and fillings.
We say that this scheme is \emph{maximal} if
we make a filling only if no expansion can be made. 
\item \emph{The decreasing scheme}. We build $\ms$ from the right to the left. Starting from $K$, we obtain $\emptyset$ by iterative collapses and perforations.
We say that this scheme is \emph{maximal} if we make a perforation only if no collapse can be made. 
\end{enumerate}

\noindent
See \cite[Section 7]{Bertrand2023MorseSequences} for a discussion of the differences between these schemes.

\begin{definition} \label{def:seq2}
The \emph{gradient vector field of a Morse sequence 
$\ms$} is the set of all regular pairs for $\ms$. 
We say that two Morse sequences $\ms$ and $\overrightarrow{V}$ on 
a given complex $K$ are \emph{equivalent} if 
they have the same gradient vector field. 
\end{definition}

It is worth mentioning that there is no loss of generality when using Morse sequences 
as a presentation of gradient vector fields. In fact, we can prove that
the gradient vector field of an arbitrary Morse function may be seen as the gradient vector
field of a Morse sequence (see \cite{Bertrand2023MorseSequences}).

Let $\ms$ be a Morse sequence on $K$.
We write $\msc = \{\sigma \in  K \; | \; \sigma$ is critical for $\ms \}$,
$\msc^{(p)} = \{\sigma \in  K^{(p)} \; | \; \sigma$ is critical for $\ms \}$, and $\msc^{(-1)} = \emptyset$. 

\noindent
For each $p \geq -1$,
we write $\msc[p]$ for the set composed of all subsets of $\msc^{(p)}$. An element $c \in \msc[p]$
is a \emph{$p$-chain of $\msc$}. We have $\msc[p] \subseteq K[p]$. 

A Morse frame is simply a map which assigns, to each $p$-simplex of $K$, a certain set of critical $p$-simplexes.

\begin{definition}
    Let $\ms$ be a Morse sequence on a simplicial complex $K$.
We say that $\Upsilon$ is a \emph{(Morse) frame on $\ms$}
if $\Upsilon$ is a map such that: 

\noindent
\hspace*{\fill}
$\Upsilon:$ $\sigma \in K^{(p)} \mapsto  \Upsilon(\sigma) \in \msc[p]$.
\hspace*{\fill} 

\noindent
If $\Upsilon$ is a Morse frame on $\ms$, we also denote by $\Upsilon$ the map:

\noindent
\hspace*{\fill}
$\Upsilon:$ $c \in K[p] \mapsto  \Upsilon(c) \in \msc[p]$, where $\Upsilon(c) = \sum_{\sigma \in c} \Upsilon(\sigma)$
and $\Upsilon(\emptyset) = 0$.
\hspace*{\fill} 
\end{definition}

\section{The Morse reference}
\label{sec:reference}
A Morse complex is a basic tool for efficiently computing simplicial homology
using discrete Morse theory. Since a Morse complex is built solely on critical complexes, its dimension 
is generally much smaller than the one of the original complex. 
In this section, we introduce two frames  which allow simplifying the construction of a Morse complex.

\subsection{Reference and co-reference}
\begin{definition} \label{def:reference}
Let $\ms$ be a Morse sequence and 
let $\Upsilon',\Upsilon''$ be two Morse frames on $\ms$
such that, for each critical simplex $\sig$ of $\ms$,
we have $\Upsilon'(\sig) = \Upsilon''(\sig) = \{ \sigma \}$. 

\noindent
We say that $\Upsilon'$ is the \emph{(Morse) reference of $\ms$} if,
for each regular pair $(\sigma,\tau)$ of~$\ms$, we have $\Upsilon'(\tau) = 0$
and $ \Upsilon'(\sig) = \Upsilon'(\partial(\tau) \setminus \{\sig \})$. 

\noindent
We say that $\Upsilon''$ is the \emph{(Morse) co-reference of~$\ms$}
if,
for each regular pair $(\sigma,\tau)$ of~$\ms$, we have 
$\Upsilon''(\sig) = 0$ and 
$ \Upsilon''(\tau) = \Upsilon''(\delta(\sig) \setminus \{\tau \})$. 
\noindent
If $\Upsilon$ is the reference of $\ms$, we write $\Upsilon^*$ for the co-reference of $\ms$.
\end{definition}
Thus, if $\Upsilon$ is the Morse reference of~$\ms$ then,
for each regular pair $(\sigma, \tau)$ of~$\ms$, 
we have $ \Upsilon(\partial(\tau)) = 0$
and $ \Upsilon^*(\delta(\sigma)) = 0$.

Let $\ms$ be a Morse sequence on $K$ and let
$\mss = \langle \hat{\sigma}_1,\ldots, \hat{\sigma}_k\rangle$. 
We see that a Morse reference $\Upsilon$ of $\ms$ may be computed by scanning the sequence $\mss$
from the left to the right. Also, a Morse co-reference $\Upsilon^*$ of $\ms$ may be computed by scanning $\mss$ 
from the right to the left. The uniqueness of $\Upsilon$ and $\Upsilon^*$ is a consequence of these constructions. 
As a limit case, observe that:
\begin{itemize}[noitemsep,topsep=0pt]
\item If $\tau$ is a facet of $K$, then we have $\Upsilon(\tau) = 0$ whenever $\tau$ is not critical.
\item If $\sig$ is a 0-simplex of $K$, then we have $\Upsilon^*(\sig) = 0$ whenever $\sig$ is not critical.
\end{itemize}
\noindent
Also, it can be checked that the
references and co-references of two Morse sequences
are equal whenever these sequences are equivalent in the sense given in Definition \ref{def:seq2}.
The converse is, in general, not true.

\begin{figure*}[tb]
    \centering
    \begin{subfigure}[t]{0.32\textwidth}
        \centering
        \includegraphics[height=.9\textwidth]{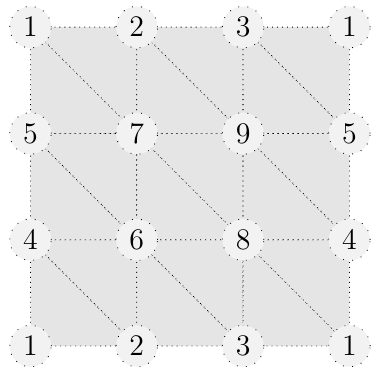}
        \caption{}
    \end{subfigure}%
    ~
    \begin{subfigure}[t]{0.32\textwidth}
        \centering
        \includegraphics[height=.9\textwidth]{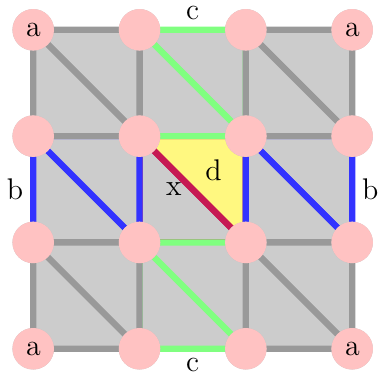}
        \caption{}
    \end{subfigure}%
     ~
    \begin{subfigure}[t]{0.32\textwidth}
        \centering
        \includegraphics[height=.9\textwidth]{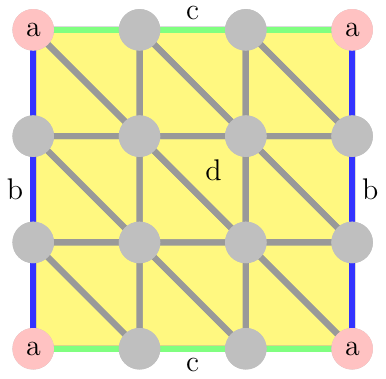}
        \caption{}                                          
    \end{subfigure}
 \caption{(a) A torus. Points with the same label are identified. (b) A Morse reference map. (c) A Morse co-reference map. See text for details.
 }
 \label{fig:MorseSequenceTorus}
\end{figure*}

Fig.~\ref{fig:MorseSequenceTorus}.a depicts a two-dimensional torus.  We first illustrate, in Fig.~\ref{fig:MorseSequenceTorus}.b, the Morse reference of a Morse sequence $\torusms$ on this torus obtained by a maximal increasing scheme (depicted in detail in~\cite[Fig.~1]{Bertrand2023MorseSequences}). In this figure, any simplex $\sigma$ in grey is such that $\Upsilon(\sigma)=0$. 
At the first step, the first critical simplex $\sigma_1=a$ is coloured in pink, with $\Upsilon(\sigma_1)=a$. After all the possible expansions from $a$, we have $\Upsilon(\sigma)=a$ (in pink) for all $\sigma$ of dimension $0$.
At the next stage, we introduce a first 1-critical simplex $b$ (in blue), and we have $\Upsilon(b)=b$; this leads to $\Upsilon(\sigma)=b$ for all simplexes $\sigma$ of dimension $1$ highlighted in blue. We then introduce a second critical 1-simplex, $c$ (in green), and we have $\Upsilon(c)=c$; this leads to $\Upsilon(\sigma)=c$ for all simplexes $\sigma$ of dimension $1$ in green. At the penultimate step of the Morse sequence, we have a free pair $(\sigma, \tau)$, and $\Upsilon(\sigma)=x$, in purple, where $x=\Upsilon(\partial(\tau)\setminus \{\sigma\})= b+c$.
The ultimate step of the Morse sequence is the critical 2-simplex $d$, and we have $\Upsilon(d)=d$, highlighted in yellow.
 
Conversely, in Fig~\ref{fig:MorseSequenceTorus}.c, by scanning $\torusms$ from right to left, we obtain its Morse co-reference map. In this figure, any simplex $\sigma$ in grey is such that $\Upsilon^*(\sigma)=0$. Starting with the critical 2-simplex $d$, after several steps in the sequence, we have $\Upsilon^*(\sigma)=d$ for all simplexes $\sigma$ of dimension 2, highlighted in yellow.
We then have $\Upsilon^*(\sigma)=c$ for all simplexes $\sigma$ of dimension 1 highlighted in green, and $\Upsilon^*(\sigma)=b$ for all simplexes $\sigma$ of dimension 1 highlighted in blue. Finally, we have $\Upsilon^*(\sigma_1)=a$ for the last simplex of dimension 0, which is critical, and is coloured in pink.

\subsection{Gradient paths, co-gradient paths and gradient flows}
References are closely related to the notion of a gradient path. 
In the following, we recall the classical definition of such a path. We also introduce the notion 
of a co-gradient path, which arises naturally from the definition of a co-reference. 

Let $\ms$ be a Morse sequence on $K$. 
 \begin{enumerate}[noitemsep,topsep=0pt]
\item Let $\pi = \langle \sigma_0, \tau_0, \ldots,\sigma_{k-1}, \tau_{k-1}, \sig_k \rangle$, $k \geq 0$,
be a sequence with $\sigma_i \in K^{(p)}$, $\tau_i \in K^{(p+1)}$.
We say that $\pi$ is a \emph{gradient path in $\ms$ (from $\sig_0$ to $\sig_k$)} if, for any $i \in [0,k-1]$,
the pair $(\sigma_i,\tau_{i})$ is regular for $\ms$ and
$\sigma_{i+1} \in \partial(\tau_{i})$, with $\sigma_{i+1} \not= \sig_{i}$.
The path $\pi$ is {\em trivial} if $k=0$, that is, if 
$\pi = \langle \sigma_0 \rangle$ with $\sigma_0 \in K^{(p)}$.
\item Let $\pi = \langle \tau_0, \sig_1,\tau_1, \ldots,\sig_{k}, \tau_{k} \rangle$, $k \geq 0$, 
be a sequence with $\tau_i \in K^{(p)}$, $\sigma_i \in K^{(p-1)}$.
We say that $\pi$ is a \emph{co-gradient path in $\ms$ (from $\tau_0$ to $\tau_k$)} if, for any $i \in [1,k]$,
the pair $(\sigma_i,\tau_{i})$ is regular for $\ms$ and
$\tau_{i-1} \in \delta(\sig_{i})$, with $\tau_{i} \not= \tau_{i-1}$. 
The path $\pi$ is {\em trivial} if $k=0$, that is, if 
$\pi = \langle \tau_0 \rangle$ with $\tau_0 \in K^{(p)}$.
\end{enumerate}

Prop.~\ref{pro:grad1} below can be proved by induction, by considering the two scanning processes of $\mss$
that are mentioned above. Theorem \ref{pro:grad11} reflects an important duality relation between the reference and the co-reference of a Morse sequence. It can be proved by changing the extremities of gradient and co-gradient paths. 

 \begin{proposition} \label{pro:grad1} Let $\ms$ be a Morse sequence on $K$
 and $\Upsilon$ be the reference of $\ms$.
 Let $\sig, \nu \in K^{(p)}$ such that $\nu$ is critical for $\ms$. 
  \begin{enumerate}[noitemsep,topsep=0pt]
 \item  We have $\nu \in \Upsilon(\sigma)$ if and only if the number of gradient paths from the simplex~$\sigma$ to the critical simplex $\nu$
is odd. 
 \item  We have $\nu \in \Upsilon^*(\sigma)$ if and only if the number of co-gradient paths from the critical simplex $\nu$ to the simplex $\sig$
is odd. 
\end{enumerate}
 \end{proposition}
 
  \begin{theorem} \label{pro:grad11} Let $\ms$ be a Morse sequence on $K$
 and $\Upsilon$ be the reference of $\ms$.
 Let $\sig \in K^{(p)}$ and $\tau \in K^{(p+1)}$
 be two simplexes that are both 
 critical for $\ms$.\\
 We have $\sig \in \Upsilon(\partial(\tau))$ if and only if $\tau \in \Upsilon^*(\delta(\sig))$.
 \end{theorem}
 
An important concept in discrete Morse theory is the one of gradient flows~\cite{For98}, by which, using Forman's own words~\cite{forman2002discretecohomology}, {\em loosely speaking, a simplex flows along the gradient paths} for infinite time (see also \cite{forman2002discretecohomology} for the dual concept). See \cite[Def. 8.6]{Sco19} for a precise definition. Gradient flows are a basic ingredient for setting the fundamental property of a Morse complex, that is, the equality of homology between a complex and its Morse complex.
 In fact, there is a deep link between co-references  of a Morse sequence and gradient flows. For the sake of space, 
 we give only an informal presentation of this relation, which may be checked by the interested readers. 
 If $\tau$ is a $p$-simplex of a complex $K$, the gradient flow which starts
 from $\tau$ is obtained: 
 \begin{enumerate}[noitemsep,topsep=0pt]
     \item By considering regular pairs $(\sig',\tau')$, with $\sig' \in \partial(\tau)$. Such a pair may be seen as the beginning of a co-gradient path that starts at $\tau$.
    \item By considering  some $p$-simplexes that are in the boundary of a $(p+1)$-simplex~$\nu$, such that $(\tau, \nu)$ is regular. %
 \end{enumerate}
If $\tau$ is a critical simplex, then the case 2. cannot happen. Also, this case cannot happen for $\tau'$, since 
$\tau'$ belongs to the regular pair $(\sig',\tau')$. By induction, 
the gradient flow starting at a critical simplex corresponds exactly to co-gradient paths. In a dual manner, 
the gradient flow ending at a critical simplex corresponds to gradient paths.
Thus, 
if $\sig, \tau \in K^{(p)}$ and if $\tau$ is critical for $\ms$, then: \\
\hspace*{\fill}
\emph{The simplex $\sig$ is in the gradient flow starting at $\tau$ if and only if 
$\tau \in \Upsilon^* (\sig)$. }
\hspace*{\fill} \\
\hspace*{\fill}
\emph{The simplex $\sig$ is in the gradient flow ending at $\tau$ if and only if 
$\tau \in \Upsilon (\sig)$. }
\hspace*{\fill}

\smallskip
It is interesting to compare $\Upsilon$ and $\Upsilon^*$ with the analogous constructions in smooth Morse theory.
A gradient flow associates a critical simplex to a chain which is invariant under the flow. According to \cite{forman2002discretecohomology}, this chain is the discrete analogue of the  {\em unstable (or descending) cell associated to a critical point of a smooth Morse function}, and it is obtained with $\Upsilon^*$. Forman \cite{forman2002discretecohomology} also studies the dual of the flow, the coflow. The coflow maps a critical simplex to a chain that is invariant under the coflow. This chain plays {\em the role of the stable (or ascending) cell associated to a critical point of a smooth Morse function}, and it is obtained thanks to $\Upsilon$.

\subsection{The Morse complex}
Now, let us consider a boundary map that is restricted to the critical simplexes. This map may be easily
built with a Morse reference.  

Let $\Upsilon$ be the reference of $\ms$.
If $\sigma \in \msc^{(p)}$, we set $d(\sigma) = \Upsilon(\partial(\sigma))$. \\
We denote by $d_p$ the map: \\
\hspace*{\fill}
$d_p:$ $c \in  \msc[p] \mapsto  d_p(c) \in \msc[p-1]$, where $d_p(c) = \Upsilon(\partial_p(c))$.
\hspace*{\fill} \\
Thus, we have 
$d_p(c) = \sum_{\sigma \in c} d(\sigma)$
with $d_p(\emptyset) = 0$.

 \begin{theorem} \label{pro:label3}
  Let $\Upsilon$ be the reference of a Morse sequence $\ms$ on $K$. \\
 For each $c \in K[p]$, we have
  $d_p (\Upsilon (c)) = \Upsilon (\partial_p (c))$.
 \end{theorem}

\begin{proof}
Let $\ms = \langle \emptyset = K_0,...,K_k =K \rangle$ be a Morse sequence on $K$, and let $\mss = \langle \hat{\sigma}_1,..., \hat{\sigma}_k\rangle$.
We consider the statement $(S_i)$: For each $c \in K_i[p]$,
 we have $d_p (\Upsilon (c)) = \Upsilon (\partial_p (c))$.
 We have $K_0[p] = \{\emptyset \}$. Thus $(S_0)$ holds. \\
 Suppose $(S_{i-1})$ holds with $0 \leq i-1 \leq k-1$.
 Let $c \in K_i[p]$. \\
1) Suppose $\hat{\sigma}_i = \sigma$, with $\sigma \in \msc$. If $\sigma \not\in c$, then we are done.
Otherwise, we have $c = c' \cup \{ \sigma \}$, with $c' \in K_{i-1}[p]$. \\
We have $\partial_p(c) = \partial_p(c') + \partial(\sigma)$.
Thus $\Upsilon (\partial_p(c)) = \Upsilon (\partial_p(c')) + \Upsilon (\partial(\sigma))$. \\
By the induction hypothesis and by the definition of $d(\sigma)$, we obtain
$\Upsilon (\partial_p(c)) = d_p (\Upsilon (c'))  + d(\sigma)$.
Therefore $\Upsilon (\partial_p(c)) = d_p (\Upsilon (c')) + d_p(\Upsilon (\{ \sigma \} )) = d_p (\Upsilon (c))$. \\
2) Suppose $\hat{\sigma}_i = (\sigma,\tau)$ is a free pair.
If $\sigma \not\in c$ and $\tau \not\in c$, then we are done. \\
2.1) Suppose $\sigma \in c$. Let $c' = c + \partial_{p+1}(\tau)$.
We have $\Upsilon (c') = \Upsilon (c) + \Upsilon (\partial_{p+1}(\tau)) = \Upsilon (c)$. We also have
$\partial_p(c') = \partial_p(c) + \partial_p(\partial_{p+1}(\tau)) = \partial_p(c)$.  \\
But
$c' = (c \setminus \{\sigma\}) + c''$,
with  $c'' = \{\eta \in \partial_{p+1}(\tau) \; | \; \eta \not= \sigma \}$. Thus $c' \in K_{i-1}[p]$. \\
By the induction hypothesis, it follows that
$d_p (\Upsilon (c')) = \Upsilon (\partial_p (c'))$. By the previous equalities, we obtain
$d_p (\Upsilon (c)) = \Upsilon (\partial_p (c))$.\\
2.2) Suppose $\tau \in c$. Let $c = c' \cup \{ \tau \}$, with $c' \in K_{i-1}[p]$. Since $\Upsilon (\tau) =0$, we obtain
$\Upsilon (c) = \Upsilon (c')$. Furthermore $\Upsilon (\partial_p (c)) = \Upsilon (\partial_p (c')) + \Upsilon (\partial (\tau))
= \Upsilon (\partial_p (c'))$.
By the induction hypothesis, we have $d_p (\Upsilon (c')) = \Upsilon (\partial_p (c'))$.
Therefore $d_p (\Upsilon (c)) = \Upsilon (\partial_p (c))$. \qed
\end{proof}

The two following results are direct consequences of Theorem \ref{pro:label3}.

 \begin{proposition} \label{pro:label1}
  Let $\Upsilon$ be the reference of a Morse sequence $\ms$ on $K$.
  For any $c,c' \in K[p]$ we have
$\Upsilon(\partial_p(c)) = \Upsilon(\partial_p(c'))$ whenever $\Upsilon (c) = \Upsilon (c')$.
 \end{proposition}

\begin{proof}
Let $c,c' \in K[p]$ with $\Upsilon (c) = \Upsilon (c')$.Thus, $d_p (\Upsilon (c)) = d_p (\Upsilon (c'))$.
By Theorem \ref{pro:label3}, we have $\Upsilon (\partial_p (c)) = \Upsilon (\partial_p (c'))$. \qed
\end{proof}

 \begin{proposition} \label{pro:label4}
If $\ms$ is a Morse sequence, then the maps $d_p$ are boundary operators. That is,
we have $d_p \circ d_{p+1} = 0$.
\end{proposition}

\begin{proof}
Let $\sigma \in \msc^{(p+1)}$.
We have $d_{p+1} (\{ \sigma \}) = d (\sigma) = \Upsilon (\partial (\sigma)) = \Upsilon (\partial_{p+1} (\{ \sigma \}))$. \\
By Theorem \ref{pro:label3}, we have
$d_p(\Upsilon (\partial_{p+1} (\{ \sigma \}))) = \Upsilon (\partial_p ( \partial_{p+1} (\{ \sigma \}))) = \Upsilon (0) = 0$. \\
Thus $d_p \circ d_{p+1} ( \{ \sigma \}) = 0$, which gives the result by linearity. \qed
\end{proof}

Since  $d_p \circ d_{p+1} = 0$, the couple $(\msc[p], d_p)$ satisfies the definition of a \emph{chain complex} \cite{Hat01}.
We say that $(\msc[p], d_p)$ is the \emph{Morse (chain) complex of $\ms$}.
This notion of a Morse complex
is equivalent to the classical one given in the context of discrete Morse theory.
This fact may be verified using \cite[Theorem 8.31]{Sco19}, Proposition \ref{pro:grad1}, and the very definition of the differential $d_p$.

Dual results for Th.~\ref{pro:label3}, Prop.~\ref{pro:label1} and Prop.~\ref{pro:label4} can be written by considering $\Upsilon^*$ instead of $\Upsilon$.
 
In the following, we denote by $H_p(\msc)$ (resp. $H^p(\msc)$) the $p^\textnormal{th}$ homology (resp. cohomology) vector space corresponding to the Morse complex of $\ms$.
By Theorem \ref{pro:label3}, the map $\Upsilon$ is a \emph{chain map} \cite{Hat01} from the chain complex $(K[p], \partial_p)$
to  the chain complex $(\msc[p], d_p)$. Hence, $\Upsilon$ induces a linear map between $H_p(K)$ and $H_p(\msc)$; see \cite{Hat01}.
Furthermore, we have the following.

\begin{theorem} [from \cite{For98}] \label{pro:label5}
For each $p \geq 0$, the vector spaces $H_p(K)$ and $H_p(\msc)$ are isomorphic.
 \end{theorem}

\section{Annotations}
\label{sec:annotations}
If $\sig$ is a $p$-simplex in a complex $K$,
an annotation for $\sig$, as introduced in  \cite{Busa12}, is a length $g$ binary vector, where $g$ is the rank of the homology group $H_p(K)$.
These annotations, when summed up for simplexes in a given cycle, provide a way to
determine the homology class of this cycle.
The following definition is an adaptation for Morse sequences of this notion.
The main difference is that we annotate each simplex with a subset of the critical simplexes of the sequence, instead of
a vector.

Let $\ms$ be a Morse sequence on $K$.
We say that a Morse frame $\Upsilon$ on $\ms$ is an  \emph{annotation on $\ms$} if $\Upsilon$ satisfies the three conditions: 
\begin{itemize}[noitemsep,topsep=0pt]
    \item [C1:] For each $\sig \in K^{(p)}$, we have $\Upsilon(\sig) \subseteq \ddot{V}^{(p)}$ where $\ddot{V}^{(p)}$ is a subset of $\msc^{(p)}$;
    \item[C2:] For each $p$, we have $Card(\ddot{V}^{(p)}) = \beta_p(K)$; 
    \item [C3:] For any cycles $z,z' \in Z_p(K)$, we have $\Upsilon(z) = \Upsilon(z')$ if and only if their homology classes are such that
$[z]_p = [z']_p$.
\end{itemize}

Let $\Upsilon$ be a frame on $\ms$. If $\tau \in \ddot{W}^{(p)}$,
we set $\Upsilon^\sharp(\tau) = \{\sigma \in K^{(p)} \; | \; \tau \in \Upsilon(\sigma) \}$. 
\noindent
The following proposition, derived from \cite{Dey14}, indicates that an annotation may be seen as a way
to determine a cohomology basis of the complex.

 \begin{proposition} [adapted from \cite{Dey14}] \label{pro:ann1}
Let $\ms$ be a Morse sequence on $K$.
A  frame $\Upsilon$ on $\ms$ is an annotation on $\ms$ if and only if $\Upsilon$ satisfies the conditions C1, C2, and the
following condition C4. 
\begin{itemize}[noitemsep,topsep=0pt]
\item [C4:] The set of chains $\{\Upsilon^\sharp(\tau) \; | \; \tau \in \ddot{V}^{(p)} \}$ is a set of cocycles whose cohomology
classes $\{[\Upsilon^\sharp(\tau)]^p \; | \; \tau \in \ddot{V}^{(p)} \}$ constitute a basis of $H^p(K)$.
\end{itemize}
\end{proposition}

We give a construction for obtaining an annotation. Again, it is an adaptation for a  Morse sequence of the one given in
\cite{Silva11} and \cite{Dey14}.
Three cases are considered: 
\begin{enumerate}[noitemsep,topsep=0pt]
\item If a critical simplex is added, and if the annotation of the boundary of this simplex is trivial, then a new cycle is created.
The label associated to this simplex is composed solely of the simplex itself. 
\item If a critical simplex is added, and if the annotation of the boundary of this simplex is not trivial, then a cycle is removed.
This is done by selecting one label in the annotation of the boundary of this simplex, and by removing this label from all
the previous annotations. 
\item If a free pair is added, we propagate the labels of the annotations to this pair, according to the simple rule of Def. \ref{def:reference}. 
\end{enumerate}
See \cite{Silva11} and \cite{Dey14} for the validity of this construction for the cases 1 and 2 The validity for the case 3 is an easy consequence
of the definition of a free pair. 

Let $\overrightarrow{W} = \langle K_0,\ldots,K_k  \rangle$ be a Morse sequence and $\widehat{W} = \langle \hat{\sigma}_1,\ldots, \hat{\sigma}_k\rangle$.
 We write $\overrightarrow{W_i} = \langle K_0,\ldots,K_i \rangle$,
 $i \in [0,k]$.
 We consider the sequence $\langle \Upsilon_{0},\ldots,\Upsilon_{i} \rangle$, $i \in [0,k]$, such that $\Upsilon_{i}$ is a frame
 for  $\overrightarrow{W_i}$, with $\Upsilon_{0}(\emptyset) = 0$ and: 
 \begin{enumerate}[noitemsep,topsep=0pt]
\item If $\hat{\sigma}_i = \sig_i$ and $\Upsilon_{i-1}(\partial(\sig_i)) = 0$, then
$\Upsilon_i$ is such that $\Upsilon_i(\sig_i) = \sig_i$ and $\Upsilon_i(\tau) = \Upsilon_{i-1}(\tau)$ otherwise. 
\item  If $\hat{\sigma}_i = \sig_i$ and $\Upsilon_{i-1}(\partial(\sig_i)) \not= 0$, then we select
an arbitrary critical face $\nu \in \Upsilon_{i-1}(\partial(\sig_i))$.
The map $\Upsilon_i$ is such that $\Upsilon_i(\sig_i) = 0$,
$\Upsilon_i(\tau) = \Upsilon_{i-1}(\tau) + \Upsilon_{i-1}(\partial(\sig_i))$ if $\nu \in \Upsilon_{i-1}(\tau)$, and
$\Upsilon_i(\tau) = \Upsilon_{i-1}(\tau)$ otherwise. 
\item If $\hat{\sigma}_i = (\sig_i,\tau_i)$, then  $\Upsilon_i$ such that $\Upsilon_i(\tau_i) = 0$,
$\Upsilon_i(\sig_i) = \Upsilon_{i-1}(\partial(\tau_i) +  \sig_i )$,
and $\Upsilon_i(\tau) = \Upsilon_{i-1}(\tau)$ otherwise. 
\end{enumerate}

Under the above construction, each frame $\Upsilon_i$ is an annotation on $\overrightarrow{W_{i}}$. 

\noindent
Let $\ddot{V}_i^{(p)} = \{ \sig \in K_i \; | \; \Upsilon_i(\sig) =  \sig  \}$, $\ddot{V}_i^{(p)}$ is composed of critical faces for $\overrightarrow{W_{i}}$. 

\noindent
For each $\sig \in K_i^{(p)}$, we have $\Upsilon_i(\sig) \subseteq \ddot{V}_i^{(p)}$.
Furthermore, for each $p$, we have $Card(\ddot{V}_i^{(p)}) = \beta_p(K_i)$.

The interested reader can check that the reference map of the torus, given in Fig.~\ref{fig:MorseSequenceTorus}.b, is indeed an annotation.
Here, the above case 2 does not happen. This case corresponds to the cell $c$ from Figure \ref{fig:MorseSequenceDunceHat}.b.

\section{Computing Betti numbers with the Morse reference}
\label{sec:mergingReferenceAnnotations}

\noindent\textbf{Note.} \emph{This section presents a corrected version of the algorithm from the published version of this paper. The correction concerns the construction described below and does not affect the overall results or conclusions.}

In the construction described in Sec.~\ref{sec:annotations}, we have to remove a label from all previous annotations. We now present another construction that reduces the amount of operations required for
this task. The basic idea is to use the information given by the reference of a Morse sequence, and to remove labels only for some faces which are in the boundary of  critical simplexes. Thus, annotations are not computed for all simplexes, but this construction allows us to obtain the Betti numbers of the complex. 

Let $\ms$ be a Morse sequence on a simplicial complex $K$, and let $\Upsilon$ be a Morse frame on $\ms$. We say that  $\Upsilon$ is {\em perfect} if each Betti number $\beta_p(K)$ is exactly equal to the number of critical $p$-simplexes $\sig$ in $\ms$
such that $\Upsilon(\sig) = \{ \sig \}$.

A key observation is the following. The Morse reference $\Upsilon$ of $\ms$ is perfect if, and only if, for any critical simplex $\sigma$ in $\ms$, we have $\Upsilon(\partial(\sigma)) = 0$. In the next construction, we take advantage of this observation to iteratively remove suitable pairs of critical simplexes from the image of $\Upsilon$.

Let $\overrightarrow{W} = \langle \emptyset = K_0,...,K_k =K  \rangle$ be a Morse sequence on a complex $K$, and let $\Upsilon$ be the Morse reference of $\ms$.
We write $\widehat{W} = \langle \hat{\sigma}_1,..., \hat{\sigma}_k\rangle$. We set: \\
\hspace*{\fill}
$\ddot{W}^{+} = \ddot{W} \cup \{\tau \in K \; | \; \tau \in \partial(\sig)$ for some $\sig \in \ddot{W} \}$ 
\hspace*{\fill} \\
We consider the sequence of frames $\Upsilon_{0},\ldots,\Upsilon_{k}$ such that
$\Upsilon_{0} = \Upsilon$ and: 
  \begin{enumerate}[noitemsep,topsep=0pt]
\item If $\hat{\sigma}_i = \sig_i$ and $\Upsilon_{i-1}(\partial(\sig_i)) = 0$, then $\Upsilon_i = \Upsilon_{i-1}$. 
\item  If $\hat{\sigma}_i = \sig_i$ and $\Upsilon_{i-1}(\partial(\sig_i)) \not= 0$, then we select
an arbitrary critical simplex $\nu \in \Upsilon_{i-1}(\partial(\sig_i))$.
The map $\Upsilon_i$ is such that: 
\begin{itemize}
\item
If $\tau \in \ddot{W}^{+}$ and  $\sigma_i \in \Upsilon_{i-1}(\tau)$, then $\Upsilon_i(\tau) = \Upsilon_{i-1}(\tau) + \sig_i$,
\item  
If $\tau \in \ddot{W}^{+}$ and  $\nu \in \Upsilon_{i-1}(\tau)$, then
$\Upsilon_i(\tau) = \Upsilon_{i-1}(\tau) + \Upsilon_{i-1}(\partial(\sig_i))$,
\item
If $\tau \not\in \ddot{W}^{+}$, then $\Upsilon_i(\tau) = \Upsilon_{i-1}(\tau)$.
\end{itemize}
\item Otherwise, $\hat{\sigma}_i$ is a free pair, we set $\Upsilon_i = \Upsilon_{i-1}$. 
\end{enumerate}

\noindent
We then have the following: the Morse frame $\Upsilon_k$ is perfect.

It is easy to check that the Morse reference $\Upsilon$ of
the torus, given in Fig.~\ref{fig:MorseSequenceTorus}.b, is such that $\Upsilon(\partial(\sigma))=0$ for all critical simplexes $\sigma$. Thus, this Morse reference is perfect, and directly gives the expected Betti numbers (1,2,1) for the torus.

\begin{figure*}[tb]
    \centering
    \begin{subfigure}[t]{0.32\textwidth}
        \centering
        \includegraphics[width=.9\textwidth]{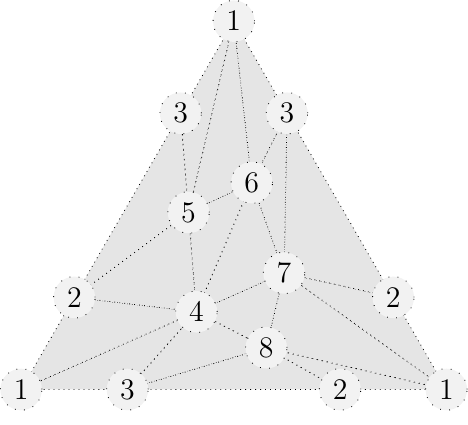}
        \caption{}
    \end{subfigure}%
    ~    
    \begin{subfigure}[t]{0.32\textwidth}
        \centering
        \includegraphics[width=.9\textwidth]{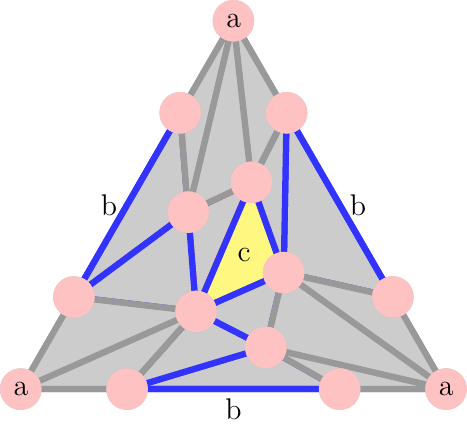}
        \caption{}
    \end{subfigure}%
     ~    
    \begin{subfigure}[t]{0.32\textwidth}
        \centering
        \includegraphics[width=.9\textwidth]{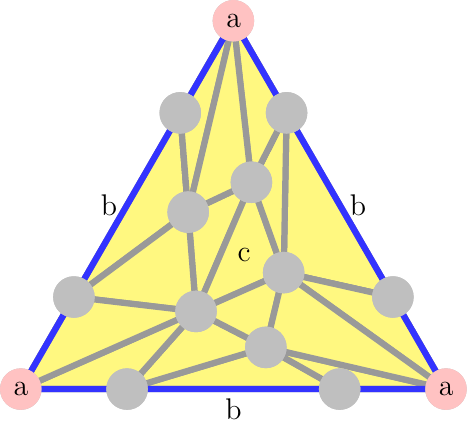}
        \caption{}
    \end{subfigure}
 \caption{(a) A dunce hat. Points with the same label are identified.  (b) A Morse reference map. (c) A Morse co-reference map. See text for details.}
 \label{fig:MorseSequenceDunceHat}
\end{figure*}
Now, let us consider the Morse reference $\Upsilon$
of the dunce hat that is depicted in Fig.~\ref{fig:MorseSequenceDunceHat}.b (see the corresponding Morse sequence in~\cite[Fig.~2]{Bertrand2023MorseSequences}). The last critical simplex $c$ in the sequence is such that $\Upsilon(\partial(c))=b+b+b=b$. 
Thus, there exist at this last step some edges annotated with $b$
in $\ddot{W}^{+}$.
We remove both $b$ and $c$ from the set of annotations of $\ddot{W}^{+}$, having for effect to “kill” the blue cocycle (or dually, to “kill” the blue cycle in Fig.~\ref{fig:MorseSequenceDunceHat}.c). We then retrieve the expected Betti numbers $(1,0,0)$ for the dunce hat. This result confirms that the dunce hat is acyclic, although the Morse sequence contains 3 critical simplexes, $a, b$ and~$c$. 

\section{Implementing Morse frames}
\label{sec:implementing}

In the literature, specific, independent algorithms are designed for computing the gradient vector field \cite{robins2011theory,benedetti2014random,harker2014discrete}, the Morse complex \cite{robins2011theory,fugacci2019computing}, or the Betti numbers~\cite{Dey14}. The framework of Morse frames shows that we can compute the gradient vector field and the Morse complex simultaneously, in only one pass, and the Betti numbers in two passes.  

As long as we can check whether a pair is free in constant time ({\em e.g.}, for example with cubical complexes and the use of a mask to check the neighborhood of a simplex), the complexity of a Morse sequence is $\mathcal{O}(dn)$, where $d$ is the dimension of the complex, and $n$  the number of its simplexes. When we compute a Morse reference map, we need to maintain a list of labels for each simplex, each label corresponding to a critical simplex. 
This leads to a complexity in $\mathcal{O}(dcn)$, where $c$ is the number of critical simplexes. The Morse reference  has a memory complexity of $\mathcal{O}(cn)$. In contrast, algorithms that compute a Morse complex, such as \cite{robins2011theory,fugacci2019computing}, claim a cubic worst-case  complexity for $d=3$ (because they have to run several times on each gradient path); furthermore, such algorithms can only be applied {\em after} obtaining a gradient vector field.

The framework of Morse frames allows retrieving the concept of annotations~\cite{Dey14}. The current implementation of annotations~\cite{boissonnat2015compressed} can be described, in our language,  as a Morse sequence where all simplexes are critical, {\em i.e.}, with only fillings. A key point for efficiency of this implementation~\cite{boissonnat2015compressed}, is the ordering of the simplexes: a heuristic is used to try preventing the creation of unnecessary cycles. Morse frames show that, with a simple change of heuristic (using, for example, a maximal increasing scheme), the annotation algorithm can take advantage of gradient fields. The ordering of the simplexes provided by such a scheme, avoids the creation of unnecessary cycles, by using expansions and fillings, instead of only fillings. %

Section~\ref{sec:mergingReferenceAnnotations} provides an algorithm for computing the Betti numbers in mod 2 arithmetic, that is inspired  by annotations. This algorithm uses the reference map and only considers the set of critical simplexes and their boundary.

\section{Conclusion}
This paper introduces Morse frames, that are based on a novel presentation of discrete Morse theory, called Morse sequences. Morse frames allow for adding information to a Morse sequence, associating a set of specific simplexes to each simplex. The main example of Morse frames, the Morse reference,
offers substantial utility in the context of homology. In particular, together with its dual, the Morse co-reference, they provide the discrete analogue of ascending/stable and descending/unstable cell associated to a critical point of a smooth Morse function. Significantly, the Morse reference allows retrieving the Morse chain complex. Using Morse frames, we give a novel presentation of the annotation algorithm. Inspired by these annotations, we describe an efficient scheme for computing Betti numbers in mod 2 arithmetic. 

On the theoretical side, for future work, we aim at providing a proof of Th.~\ref{pro:label5}, that will rely only on the Morse reference. We also intend to compute persistence with the Morse reference, and to extend our framework to other fields than the mod 2 arithmetic. On a more practical level, we also want to test the proposed algorithms, and to compare their efficiency with respect to the state-of-the-art. 

\newpage
\bibliographystyle{splncs04}
\bibliography{biblio}

\end{document}